\documentclass{article}
\usepackage{mathtools,paralist,subcaption,amsthm}
\usepackage[natbib,backend=bibtex]{biblatex}
\usepackage{dmbiblatex}
\usepackage{tikz}
\usetikzlibrary{shapes.geometric,automata,positioning,arrows}
\usepackage{hyperref}
\bibliography{pushdown_cache}

\title{The complexity gap in the static analysis of cache accesses grows if procedure calls are added}
\author{David Monniaux}

\newcommand{\pushtrans}[3]{{#1} \xhookrightarrow{#2} {#3}}
\newcommand{\pushtransstar}[3]{{#1} \xhookrightarrow{#2}^* {#3}}
\newcommand{\emptyword}{\varepsilon}

\newcommand{\true}{\mathbf{1}}
\newcommand{\false}{\mathbf{0}}

\newcommand{\nWays}{K}
\newcommand{\nRegs}{r}
\newcommand{\noAccess}{\varepsilon}

\newcommand{\JACM}{\cite{Monniaux:2019:CCA:3368192.3366018}}
\newcommand{\JACMt}{\citet{Monniaux:2019:CCA:3368192.3366018}}

\newtheorem{theorem}{Theorem}
\newtheorem{lemma}{Lemma}
\newtheorem{corollary}{Corollary}
\theoremstyle{definition}
\newtheorem{definition}{Definition}
\newtheorem{example}{Example}

\begin{document}
\maketitle

\begin{abstract}
  The static analysis of cache accesses consists in correctly predicting which accesses are hits or misses.
While there exist good exact and approximate analyses for caches implementing the least recently used (LRU) replacement policy, such analyses were harder to find for other replacement policies.
A theoretical explanation was found: for an appropriate setting of analysis over control-flow graphs, cache analysis is PSPACE-complete for all common replacement policies (FIFO, PLRU, NMRU) except for LRU, for which it is only NP-complete.

In this paper, we show that if procedure calls are added to the control flow, then the gap widens: analysis remains NP-complete for LRU, but becomes EXPTIME-complete for the three other policies.
For this, we improve on earlier results on the complexity of reachability problems on Boolean programs with procedure calls.

In addition, for the LRU policy we derive a backtracking algorithm as well as an approach for using it as a last resort after other analyses have failed to conclude.

\end{abstract}

\section{Introduction}
Most processors, except the smallest ones, implement some form of \emph{caching}: fast memory close to the CPU core retains frequently accessed code and data, to avoid slow access to external memory.
A hardware cache is split into \emph{cache sets}, and a given memory block may be stored only in a certain cache set, depending on its address.
A cache set contains $\nWays$ blocks, where $\nWays$ is known as the \emph{associativity} or \emph{number of ways}.
When a new block is loaded into a cache set, a \emph{cache replacement policy} determines which block is to be evicted to make room for it.
The most intuitive cache replacement policy is to evict the \emph{least recently used} (LRU) block. However, due to difficulties in implementing that policy efficiently in hardware, other policies with supposedly ``close'' behavior (PLRU, NMRU) are often used instead; sometimes the simple \emph{first-in first-out} (FIFO) policy is used.

Loading data from main memory, or from a more distant level of cache, may take $6$ to $100$ times the time taken for loading it from a close cache.
Not only does it directly affect execution time, it also results, especially in processors with out-of-order execution, in different microarchitectural execution patterns, themselves having an impact on execution time.
Static analysis tools used to compute bounds on worst-case execution time%
\footnote{Absint's aiT is one such tool, used in industries such as avionics, automotive, energy and space. \url{https://www.absint.com/ait/}
  Non-commercial tools include OTAWA. \url{http://www.otawa.fr/}}
thus include a static analysis for caches, meant to predict which accesses are always cache \emph{hits} (data in cache) and which are always \emph{misses} (data not in cache).%
\footnote{Static analysis tools may perform more refined analyses, such as persistence analysis, refinements according to execution paths or loop indices, etc. We do not cover these here. Our goal is to study difficulty even in the simplest, most easily understood analysis.}

\begin{figure}
\begin{subfigure}{0.45\textwidth}
\begin{center}
    \begin{tikzpicture}[node distance=4em,->,auto]
      \node (start) [diamond,draw] {  };
      \node (q0) [right of=start] { $v_0$ };
      \node (q1) [above right of=q0] { $v_1$ };
      \node (q3) [below right of=q1] { $v_3$ };
      \node (q2) [below right of=q0]  { $v_2$ };
      \path (start) edge node {$a$} (q0);
      \path (q0) edge node { $a$ } (q1);
      \path (q2) edge node { $c$ } (q0);
      \path (q1) edge node { $b$ } (q2);
      \path (q1) edge node { $d$} (q3);
      \path (q3) edge node { $e$} (q2);
    \end{tikzpicture}
\end{center}
\caption{Original control-flow graph for two cache sets:
$\{a,e\}$ and $\{b,c,d\}$.}

\label{fig:CFG}
\end{subfigure}
\hfill
\begin{subfigure}{0.45\textwidth}
\begin{center}
    \begin{tikzpicture}[node distance=4em,->,auto]
      \node (start) [draw,diamond] {  };
      \node (q0) [right of=start] { $v_0$ };
      \node (q1) [above right of=q0] { $v_1$ };
      \node (q3) [below right of=q1] { $v_3$ };
      \node (q2) [below right of=q0]  { $v_2$ };
      \path (start) edge node {$a$} (q0);
      \path (q0) edge node { $a$ } (q1);
      \path (q2) edge node { $\noAccess$ } (q0);
      \path (q1) edge node { $\noAccess$ } (q2);
      \path (q1) edge node { $\noAccess$} (q3);
      \path (q3) edge node { $e$} (q2);
    \end{tikzpicture}
\end{center}
\caption{The same control-flow graph sliced for cache set $\{a,e\}$.}
\end{subfigure}

\caption{Slicing of a control-flow graph according to a cache set, from \citep{Touzeau:2019:FEA:3302515.3290367}. $a,b,c,d,e$ are identifiers of memory blocks. $\noAccess$ means that no cache access takes place along that edge.}
\label{fig:CFG-slicing}
\end{figure}
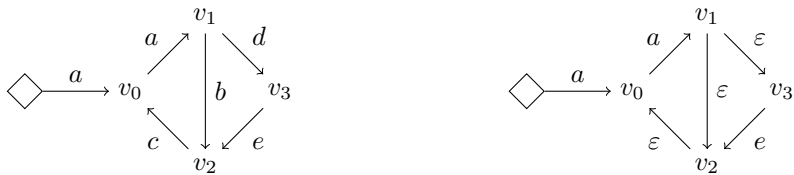

For almost all policies (including LRU, FIFO, PLRU and NMRU) found in processors, analysis for exist-hit and exist-miss properties may be performed on each cache set separately, by slicing the program according to the cache set (\autoref{fig:CFG-slicing}), without loss of precision (the only policy for which this is false is pseudo-round-robin, which we do not consider here)~\citep{Touzeau:2019:FEA:3302515.3290367}.%
\footnote{This could be incorrect if we were considering complex microarchitectures with cache prefetching etc., since the availability of data in a cache set may result in loads being made or not made to other cache sets. Again, we consider a simple setting here. Separate analysis may however be used for safe overapproximations of the behavior of the system.}
We shall thus consider in this paper, without loss of generality, that the cache consists of a single cache set.

Consider the simple setting where a program is defined by a control-flow graph with edges adorned by identifiers of memory blocks ($a$, $b$\dots) meaning that that block is read on that edge (\autoref{fig:CFG-slicing})---we consider the simple case of read-only caches; caches with writes add some complications.
Equivalently, we can see it as a finite automaton that accepts a sequence of block accesses.
Two interesting decision problems are: at one control vertex, is block $x$ always in the cache (\emph{always hit}), whatever the path taken from the entry point of the graph, assuming an initially empty cache? Is it always not in cache (\emph{always miss})?
These two problems are often grouped into one: classify blocks into ``always hit'', ``always miss'', and ``unknown''.
Equivalently, by negating the problems, one may consider \emph{exist hit}: ``does there exist a path so that the block is in the cache?''; \emph{exist miss}: ``does there exist a path so that the block is not in the cache?''.
Similar problems exist for an initially unknown cache, with quantification over all paths and over all initial cache contents.
All these problems are decidable, if only by enumerating all reachable cache states (there is a finite number of cache blocks in the problem).

Approximate static analyses affix the ``unknown'' classification to blocks that may actually be ``always hit'' or ``always miss'', but the analysis is too coarse to notice it.
In contrast, an exact static analysis affixes the ``unknown'' classification only to blocks whose cache status is \emph{definitely unknown}: there exist different paths (or, if applicable, different initial cache contents) such that the block is in the cache for one and out of the cache for another.
For complexity-theoretic results, we only consider exact analyses.%
\footnote{A safe, constant-time, approximate static analysis is to answer ``unknown'' to any request. In order to study complexity, some form of minimal precision must be imposed. It is unclear what metric should be used for this; thus our choice to require exactness.}

By encoding the behavior of the cache into the program itself, seeing memory accesses as actions on the cache contents seen as a vector of bits, one turns the exact analysis ``exist hit'' and ``exist miss'' problems into model-checking reachability problems, solvable in PSPACE. If the control-flow graph is acyclic, similar reasoning leads to membership in NP.
{\JACMt}, considering the LRU, PLRU, NMRU and FIFO policies, showed, in addition, that:
\begin{itemize}
\item for all these policies (for NMRU, only with an initially empty cache), the analysis problems for acyclic control-flow graphs are NP-hard;
\item for all these policies except LRU (for NMRU, only with an initially empty cache), the analysis problems for general control-flow graphs are PSPACE-hard;
\item for LRU, all analysis problems (regardless of initial state or acyclicity) are NP-complete.
\end{itemize}
To summarize, for all policies, the analysis problems are NP-complete for acyclic control-flow graphs, but there is a gap between LRU (NP-complete) and the others (PSPACE-complete) for general control-flow graphs.
{\JACMt} however left to future work the question of adding procedure calls (pushdown control) to the setting.

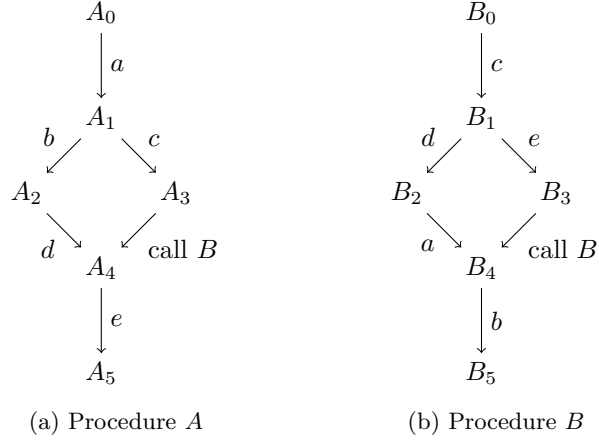
\begin{figure}\hfil
  \begin{subfigure}{0.25\textwidth}
  \begin{tikzpicture}[node distance=4em,->,auto]
    \node (A0) { $A_0$ };
    \node (A1) [below of=A0] { $A_1$ };
    \node (A2) [below left of=A1] { $A_2$ };
    \node (A3) [below right of=A1] { $A_3$ };
    \node (A4) [below right of=A2] { $A_4$ };
    \node (A5) [below of=A4] { $A_5$ };
    \path (A0) edge node { $a$ } (A1);
    \path (A1) edge node [above left] { $b$ } (A2);
    \path (A1) edge node { $c$ } (A3);
    \path (A2) edge node [below left] { $d$ } (A4);
    \path (A3) edge node { call~$B$ } (A4);
    \path (A4) edge node { $e$ } (A5);
  \end{tikzpicture}
  \caption{Procedure $A$}
  \end{subfigure}\hfil
  \begin{subfigure}{0.25\textwidth}
  \begin{tikzpicture}[node distance=4em,->,auto]
    \node (B0) { $B_0$ };
    \node (B1) [below of=A0] { $B_1$ };
    \node (B2) [below left of=A1] { $B_2$ };
    \node (B3) [below right of=A1] { $B_3$ };
    \node (B4) [below right of=A2] { $B_4$ };
    \node (B5) [below of=A4] { $B_5$ };
    \path (B0) edge node { $c$ } (B1);
    \path (B1) edge node [above left] { $d$ } (B2);
    \path (B1) edge node { $e$ } (B3);
    \path (B2) edge node [below left] { $a$ } (B4);
    \path (B3) edge node { call~$B$ } (B4);
    \path (B4) edge node { $b$ } (B5);
  \end{tikzpicture}
  \caption{Procedure $B$}
  \end{subfigure}\hfil

  \caption{Two procedures, including a recursive one. Running procedure $A$ may yield traces $abde$, $accdabe$, $accecdabbe$\dots}
  \label{fig:procedures}
\end{figure}

In this paper, we show that \emph{the gap widens when procedure calls are added}: the decision problems remain NP-complete for LRU for programs with procedure calls (\autoref{fig:procedures}), but become EXPTIME-complete for other policies (for NMRU, EXPTIME hardness is proved only with an initially empty cache).%
\footnote{For complexity theoretical purposes, we assume that the input is the program to be analyzed, as a set of procedures consisting of explicitly represented control-flow graphs labeled with array accesses, preceded by the associativity $\nWays$ of the cache written in unary notation.}

For the LRU policy, we derive backtracking algorithms that solve the exist-miss and exist-hit problems from the arguments of the proof of membership in NP.
These algorithms may be costly, but we give an approach for using them only when some cheaper analyses have failed to conclude on whether an access is always a hit, always a miss, or has ``definitely unknown'' status.

As a secondary contribution, we provide an alternate proof (\autoref{part:exptime_complete}) that reachability in Boolean programs with procedure calls is EXPTIME-complete, a fact already shown by \citet{DBLP:conf/tacas/GodefroidY13}.

\section{Pushdown systems}
At several points, we shall use classical results on reachability in pushdown systems, which are used to model programs without variables (only control locations) but with procedure calls. The set of cache blocks will be the alphabet.

A \emph{pushdown system} is a quadruple $(Q,\Sigma,\Gamma,\Delta)$ where $Q$ is a finite set of \emph{control locations}, $\Sigma$ is a finite \emph{word alphabet}, $\Gamma$ is a finite \emph{stack alphabet}, 
$\Delta \subseteq (Q \times \Gamma) \times \big(\Sigma \cup \{ \emptyword \}\big) \times (Q \times \Gamma^*)$ is a finite set of \emph{transition rules}.
A pair $(q,w) \in Q \times \Gamma$ is a \emph{configuration} of the automaton.
If $\big((q,\gamma),\sigma,(q',w)\big) \in \Delta$, then we write
$\pushtrans{(q,\gamma)}{\sigma}{(q',w)}$.%
\footnote{This is the same definition as \cite{DBLP:conf/cav/EsparzaHRS00} except we keep a word alphabet.}
We write $\pushtransstar{c_0}{\sigma_1 \dots \sigma_n}{c_n}$ for
$\exists c_1,\dots,c_{n-1}~\pushtrans{c_0}{\sigma_1}{c_1} \land \dots \land
\pushtrans{c_{n-1}}{\sigma_n}{c_n}$.

To define what it means for a word to be \emph{accepted} by a pushdown system, we add to the quadruple two more items: $q_0$ an \emph{initial state}, $Q_f \subseteq Q$ a set of \emph{final states}.
A word $w=w_1\dots w_n$ is \emph{accepted with arbitrary final stack} if there is a final state $q_f \in Q_f$, a final stack $\gamma_f$, a sequence of configurations $c_0,\dots,c_n$ and transitions $\pushtrans{c_0}{w_1}{c_1}$, \dots, $\pushtrans{c_{n-1}}{w_n}{c_n}$ such that $c_0=(q_0,\emptyword)$ and $c_f=(q_f,\gamma_f)$.

A configuration $(q,w)$ is a word (over $Q \cup \Gamma$) thus we talk of sets of configurations recognized by finite automata.
\cite{DBLP:conf/concur/BouajjaniEM97}, \citet[Th.~3 (respectively, 1)]{DBLP:conf/cav/EsparzaHRS00} proposed algorithms that compute, in polynomial time, finite automata that recognize the set of reachable (respectively, co-reachable) configurations from a set of configurations defined by a finite automaton.
The following classical result ensues:

\begin{theorem}\label{th:reachability_pushdown_polynomial}
  There is a polynomial-time algorithm that,
  given an explicitly represented pushdown system and
  two regular sets $I$ and $F$ of configurations
  (represented using finite automata),
  checks whether there are configurations $i \in I$ and $f \in F$ and a word
  $w \in \Sigma^*$ such that $\pushtransstar{i}{w}{f}$.
\end{theorem}

By ``explicitly represented'' we mean that the states and the transitions are explicitly enumerated in the input.

\section{Least recently used policy}
The least recently used policy operates as follows: a $\nWays$-way cache ($\nWays$ is the \emph{number of ways} or \emph{associativity}) is a list of at most $\nWays$ distinct cache blocks. Blocks are taken from the finite set of blocks accessed by the program under analysis.
The \emph{age} of a block is its position in the list: $0$ for the most recently used block, $\nWays-1$ for the least recently used block, and, by convention, $\infty$ for blocks not in the cache.
Depending on the situation being modeled, this list may be taken initially empty (empty initial cache) or may take any initial value (arbitrary initial cache).

When a block $x$ is accessed, if it belongs to the current cache state then it is moved to the foremost position in the list, and other blocks keep their relative order; if it does not, then it is put at the foremost position and the block of age $K-1$ is evicted (if there is such a block).
The block of age $K-1$ is the least recently used in the cache, thus the name of the policy.
For instance, for a 4-way cache containing initially $abcd$, after an access to $b$ the cache contains $bacd$, and if instead $e$ is accessed the cache then contains $eabc$.

\subsection{NP membership of the analysis problem}
\label{sec:LRU_in_NP}
The simple observation at the basis of all methods for static analysis of LRU caches \citep{Ferdinand99,Touzeau:2019:FEA:3302515.3290367,Touzeau_et_al_CAV2017} is that 
on a $\nWays$-way LRU cache, a block $a$ may be in the cache after a finite execution $e$ if and only if one of these conditions is met:
\begin{itemize}
\item if starting from an arbitrary cache state: if fewer than $\nWays$ distinct cache blocks have been accessed along $e$;
\item regardless of the initial cache state, if fewer than $\nWays$ distinct cache blocks have been accessed along $e$ since the last access to $a$.
\end{itemize}


In {\JACM}, the proofs that the exist-hit and exist-miss problems for control-flow graphs adorned with cache accesses under the LRU policy could be solved in NP relied on path compression: if there exists an execution of arbitrary length that reaches a control location $l$ with $a$ in the cache (respectively, not in the cache), then there exists one of polynomial length with the same property, which is the NP witness.
More specifically, the proof relied on the possibility to ``compress'' an execution path between two control locations into a path of polynomial length with the same set of control edges along the two paths, using classical arguments such as ``if the length of an execution exceeds the number of states, then it encounters the same state twice''.

Adding procedure calls to a finite automaton classically turns it into a pushdown system.
Unfortunately, we cannot expect the same kind of compression results with pushdown systems.
Indeed, if a finite automaton accepts a word then it accepts a word of length bounded by its number of states, but the same does not apply to pushdown systems:
consider a program composed of $n$ procedures $f_0,\dots,f_n$, with $f_i$ just making two successive calls to $f_{i+1}$ for $i < n$, and $f_n$ executing an instruction $a$, then the shortest (and only) execution of $f_0$ executes $2^n$ instructions~$a$, with an exponential number of different configurations (recall that a configuration consists in a control location and a call stack).

Instead, we will use a weaker property. The witness for the existence of an execution path of arbitrary length is the sequence of first occurrences of block accesses along that path.
The crux of the argument is that it is possible, given such a sequence, to check in polynomial time for the existence of an execution path matching that sequence.

\begin{definition}
  Let $w$ be a word over an alphabet $\Sigma$, let $F(w)$ denote the sequence of first occurrences of letters in $w$.
  For instance, $F(dadaaabbaaabcbbaa) = dabc$.
\end{definition}

Let us first recall the following classical result, obtained through an automaton product construct:
\begin{theorem}\label{th:pushdown_regular}
  Let $A$ be a pushdown system over alphabet $\Sigma$ and let $A'$ be a finite automaton (deterministic or not) over~$\Sigma$.
  Then the intersection of the languages recognized by $A$ and $A'$ is recognized by a pushdown system with control locations in the Cartesian product of the control locations of $A$ and $A'$.
\end{theorem}

We shall use variants of the following lemma:
\begin{lemma}\label{th:pushdown_sequence_nonempty}
  Let $A$ be a pushdown system over alphabet $\Sigma$.
  Let $s$ be a sequence of pairwise distinct letters in $\Sigma$.
  Then it is possible to check in polynomial time for the existence of a word $w$ accepted (with arbitrary final stack) by $A$ such that $F(w)=s$.
\end{lemma}

\begin{proof}
  This problem is equivalent to testing for the nonemptiness of the intersection of the language recognized by $A$ and the language recognized by the regular expression $Z(s) = s_1 s_1^* s_2 (s_1 | s_2)^* \dots s_n (s_1 | s_2 | \dots | s_n)^*$.
  That latter language is recognized by the automaton
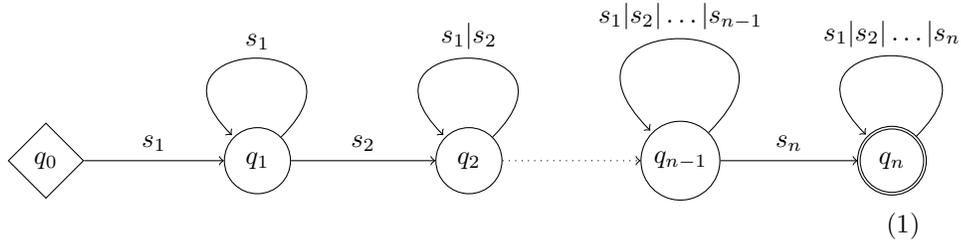
\vspace{-2em}
  
  \begin{equation}\label{equ:automaton}
  \begin{tikzpicture}[node distance=8em,->,auto]
    \node(q0) [state,diamond] {$q_0$};
    \node(q1) [state,right of=q0] { $q_1$ };
    \node(q2) [state,right of=q1] { $q_2$ };
    \node(qnm1) [state,right of=q2] { $q_{n-1}$ };
    \node(qn) [state,accepting,right of=qnm1] { $q_n$};
    \path(q0) edge node {$s_1$} (q1);
    \path(q1) edge node {$s_2$} (q2);
    \path(q2) edge [dotted] (qnm1);
    \path(qnm1) edge node {$s_n$} (qn);
    \path(q1) edge[loop] node[above] {$s_1$} (q1);
    \path(q2) edge[loop] node[above] {$s_1|s_2$} (q2);
    \path(qnm1) edge[loop] node [above]{$s_1|s_2|\dots|s_{n-1}$} (qnm1);
    \path(qn) edge[loop] node[above] {$s_1|s_2|\dots|s_n$} (qn);
  \end{tikzpicture}
\end{equation}

By \autoref{th:pushdown_regular} the intersection is recognized by a pushdown system with $(n+1)|A|$ control locations where $|A|$ is the number of control locations in~$A$.
Its emptiness can be checked in polynomial time (\autoref{th:reachability_pushdown_polynomial}).
\end{proof}

\begin{theorem}
  The exist-hit and exist-miss problems can be solved in NP for both an empty initial cache and an arbitrary initial cache.
\end{theorem}

\begin{proof}
  Let $A$ be the pushdown system defining the program to be analyzed.
  
  A block $a$ is in the cache at a control location $q$ during an execution starting from the empty cache if and only if during this execution there is a transition over letter $a$ followed by transitions over $n < \nWays$ pairwise distinct letters;
  in other words, if there exists $s$, $|s| < \nWays$ such that $q$ is reached by an execution of $A$ for a word matching the regular expression $?^* a Z(s)$ where $Z(s)$ is as in the proof of \autoref{th:pushdown_sequence_nonempty}, and $?$ matches any block.
  This is solvable in NP by guessing $s$ and then similarly as in the proof of~\autoref{th:pushdown_sequence_nonempty}.

  A block $a$ is outside the cache at a control location $q$ during an execution starting from the empty cache if and only if either this execution reaches $q$ after using only labels different from $a$, or it reaches $q$ after at least one occurrence of $a$ followed by at least $n \geq \nWays$ pairwise distinct letters distinct from~$a$.
  The first case is obtained by checking the reachability of $q$ in the restriction of $A$ to letters different from~$a$, in polynomial time (\autoref{th:reachability_pushdown_polynomial}).
  The second case is obtained by guessing a sequence $s_1,\dots,s_n$ of pairwise different accesses distinct from $a$ and again using a variant of the proof of~\autoref{th:pushdown_sequence_nonempty}, where the final state of the finite automaton loops onto itself with any letter (allowing for more letters than $s_1,\dots,s_n$).

  There exists an initial cache state and an execution such that $a$ is in the cache at control location $q$ if and only if there exists a sequence $s=s_1,\dots,s_n$ of distinct letters with $n < \nWays$ such that one reaches $q$ after reading a word $w$ such that $F(w)=s$, or there exists a sequence $s=s_1,\dots,s_n$, $n < \nWays$, and words $w_1,w_2$ such that $w_1 w_2$ leads to $q$ and $F(w_2) = a,s_1,\dots,s_n$. Both cases can be checked in NP, as in previous paragraphs.
  
  There exists an initial cache state and an execution such that $a$ is not in the cache at control location $q$ if and only if there exists a sequence reaching $q$ by going only through letters distinct from $a$, which can be checked in polynomial time, or there exists a sequence of pairwise distinct letters also distinct from $a$ $s=s_1,\dots,s_n$, $n \geq \nWays$, and words $w_1,w_2$ such that $w_1 w_2$ leads to $q$ and $F(w_2) = a,s_1,\dots,s_n$, which can be checked in NP, as in previous paragraphs.
\end{proof}

\begin{corollary}
  The exist-hit and exist-miss problems are NP-complete for both an empty initial cache and any initial cache.
\end{corollary}

\begin{proof}
  These problems subsume the corresponding problems for finite (non pushdown) automata, which were proved to be NP-hard by~{\JACM}. 
\end{proof}

\subsection{Backtracking algorithm}
\label{sec:backtracking}
We shall now see how to exploit the NP structure of the LRU analysis decision problem to build backtracking algorithms. The question we address is: given a program possibly consisting of multiple procedures, and a control edge $q_1 \xrightarrow{a} q_2$ labelled with an access to block $a$ within that program, answer whether this access is unreachable, always-hit, always-miss, or has ``definitely unknown'' status, for the LRU policy.
This problem has two variants depending on whether the initial cache is assumed to be empty or arbitrary.

We have seen that membership in NP is established by ``guessing'' a sequence of newly seen blocks in their order of appearance, whence a reachability problem for a pushdown system is created and checked in polynomial time.
The backtracking algorithm will explore possible sequences, and chronologically backtrack when it notices it has entered a search branch that cannot lead to reachability.

Let us begin with the ``exist-hit'' from the initial empty cache subproblem: for an access $q_1 \xrightarrow{a} q_2$ in the pushdown program, find if there is an execution such that this access is in the cache starting from an empty cache; in other words, if there is an execution from the program start to a first access to $a$, then accesses to fewer than $\nWays$ distinct blocks, then $q_1 \xrightarrow{a} q_2$.
This is equivalent to computing the set of all configurations occurring just after accessing $a$ and reachable from the initial state, then checking if there is an execution with accesses to fewer than $\nWays$ distinct blocks from one such configuration to control location~$q_1$.
The reachable configurations, a regular set, can be computed in polynomial time \cite{DBLP:conf/concur/BouajjaniEM97};
the question is to find the remainder of the execution, accessing a sequence  $s=s_1 \dots s_n$ of fewer than $\nWays$ distinct blocks.
For this, we conduct a backtracking search over $s$.
Similarly, the problem with arbitrary initial cache reduces to that backtracking search.

Let us now cast the results of Section~\ref{sec:LRU_in_NP} in a more effective light in order to derive a backtracking algorithm.
Consider Theorem~\ref{th:pushdown_sequence_nonempty} and its proof, which, given a sequence of letters $s$, constructs the product $P(s)$ of the finite automaton (\ref{equ:automaton}) with the pushdown system.
Remark that this product can be constructed iteratively as $s$ is extended: when a new letter is added to $s$, new states are added to the product.
Then, we compute the regular set of configurations reachable by this automaton, using the algorithm from \citet{DBLP:conf/concur/BouajjaniEM97}.
Again, if a new letter $l$ is added to $s$, and we need to compute the reachable states for the product automaton $P(sl)$ then we can retain the reachable configurations pertaining to the states already in $P(s)$ and just compute the ones for the new states.
If no reachable configurations are found in these new states, the search stops exploring $sl$ and backtracks.
If a ``winning'' state (a state with control location $q_1$ in the pushdown system) is found with those states, the search terminates.
All letters $l$ are tried, except for~$a$; the search goes recursively with prefix~$sl$.
Sequences of length greater than or equal to $\nWays$ are not explored.

The backtracking algorithm may be modified to look for sequences that lead to the eviction of~$a$, to solve the ``exist-miss'' problem.
Let $B$ be the original pushdown system from which edges labeled $a$ are removed; we pre-compute the configurations co-reachable from $q_1$ in that pushdown system.
We then look for a block sequence $s$ such that $|s| = \nWays$. A winning state $(p,f)$ in the product pushdown system, where $p$ is a control state in the original pushdown system and $f$ is a state of the finite automaton, is one where the set of its reachable configurations intersects with the set of configurations at $p$ co-reachable from $q_1$ in~$B$.
  In other words, there is at least one configuration reachable by a sequence of distinct letters of length at least $\nWays$, and co-reachable using transitions labeled by letters other than $a$ from $q_1$, which establishes the existence of a sequence of at least $\nWays$ distinct letters, distinct from $a$, reaching~$q_1$.
  Sequences of length greater than $\nWays$ need not be explored.

A further improvement to both algorithms is to systematically intersect the sets of reachable configurations with the co-reachable sets of configurations from $B$, since we are interested only in configurations that can ultimately lead to~$a_1$.

\subsection{Combination with other algorithms}
\citet{Touzeau:2019:FEA:3302515.3290367} proposed using their exact analysis as a last resort when other, cheaper analyses, could not resolve the analysis problem, focusing it on the unresolved accesses. In this section, we take this approach one step further.

They proposed using as a first step an analysis tracking the possible ages of the blocks \citep{Touzeau_et_al_CAV2017}, improving upon the well-known age interval analysis proposed by~\citet{Ferdinand99}.
That age interval analysis computes, at every location and for every block $a$, an upper bound $h_a$ and a lower bound $l_a$ on the age of that block in the cache ($+\infty$ denoting a block outside the cache), whatever the execution.
This analysis can prove that some blocks are always in the cache, or outside of the cache, at a given location.
\citet{Touzeau_et_al_CAV2017} improve upon that analysis by considering, for every block $a$ and every reachable location, the set $S$ of ages for $a$ reachable by all executions at that location, and computing four bounds $l_a$, $l'_a$, $h'_a$, $h_a$ such that $l_a \leq \min S \leq l'_a$ and $h'_a \leq \max S \leq h_a$;
in other words, there is at least one execution that reaches that location with the age for $a$ at most $l'_a$ and one execution that reaches it with that age at least $h'_a$.
This may establish cheaply that the status of some blocks is ``definitely unknown'' at some locations, meaning that there exist executions for which they are in the cache and some in which they are not.
The more expensive exact analysis is called only when the age-based analyses cannot conclude that a block must be in the cache at that location, must be outside of the cache, or has cache status ``definitely unknown''.

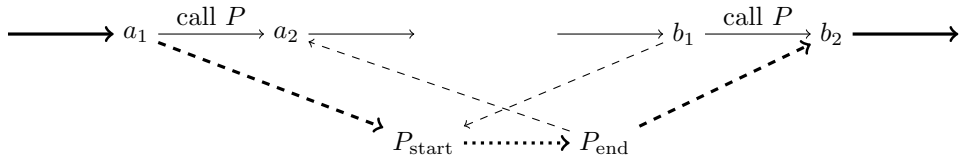
\begin{figure}
  \begin{center}
  \begin{tikzpicture}[node distance=1em and 4em, ->,auto]
    \node(a0) {};
    \node(a1) [right= of a0] {$a_1$};
    \node(a2) [right= of a1] {$a_2$};
    \node(a3) [right= of a2] {};
    \node(b0) [right= 4em of a3] {};
    \node(b1) [right= of b0] {$b_1$};
    \node(b2) [right= of b1] {$b_2$};
    \node(b3) [right= of b2] {};

    \node(Pstart) [below= 3em of a3] {$P_{\mathrm{start}}$};
    \node(Pend) [right= 4em of Pstart] {$P_{\mathrm{end}}$};

    \path(a0) edge[very thick] (a1);
    \path(a1) edge node[above]{call $P$} (a2);
    \path(a2) edge (a3);

    \path(b0) edge (b1);
    \path(b1) edge node[above]{call $P$} (b2);
    \path(b2) edge[very thick] (b3);

    \path(Pstart) edge[dotted,very thick] (Pend);

    \path(a1) edge[dashed,very thick] (Pstart);
    \path(Pend) edge[dashed] (a2);

    \path(b1) edge[dashed] (Pstart);
    \path(Pend) edge[dashed,very thick] (b2);
  \end{tikzpicture}
  \end{center}
  \caption{``Flattened'' control-flow includes paths that cannot be executed in the real system. The two ``call $P$'' edges are replaced by call and return edges (dashed). The path (thick lines)
    $a_1 \rightarrow P_{\mathrm{start}} \dots P_{\mathrm{end}} \rightarrow b_2$
    cannot be executed on the real system, because the call and return sites do not match, but exists in the ``flattened'' system.
    The ``flattened'' system thus strictly over-approximates the original behaviors.}
  \label{fig:flattened}
\end{figure}

It is possible to run these analyses on programs with procedures calls by ``flattening'' the structure, abstracting these calls by replacing each edge from $l$ to $l'$ labeled with a call by one \emph{call edge} going from $l$ to the start of the called procedure, and a \emph{return edge} going from the end of that procedure back to~$l'$.
If a procedure is called only from one location, this amounts to inlining that procedure at the point of call, and does not introduce any extra abstraction.
However, if a procedure is called from multiple locations, this introduces spurious execution traces. Suppose $P$ is called from two places $a_1 \xrightarrow{\text{call~}P} a_2$ and $b_1 \xrightarrow{\text{call~}P} b_2$.
The flattened control structure allows executions $a_1 \xrightarrow{\text{call~}P} b_2$ and $b_1 \xrightarrow{\text{call~}P} a_2$, which are impossible in the original program (Fig.~\ref{fig:flattened}).
Upper and lower bounds on ages $l_a$ and $h_a$ computed on the flattened structure are sound, meaning that are valid for the original pushdown structure, but the $l'_a$ and $h'_a$ bounds, as well as the exact analysis, are not necessarily valid because they may reflect executions that exist in the flattened structure but not in the pushdown structure.
For instance, if, in the above example, a block $a$ is always in the cache when reaching $a_1$, but never in the cache when reaching $b_1$, and $P$ is such that it it neither evicts $a$ nor accesses it in both calls, then $a$ is always in the cache when reaching $a_2$ and never in the cache when reaching $b_2$.
However, in the flattened control structure, the spurious execution $a_1 \xrightarrow{\text{call~}P} b_2$ leads to $a$ in the cache at position $b_2$, and the spurious execution $b_1 \xrightarrow{\text{call~}P} a_2$ leads to $a$ out of the cache at position~$a_2$.

If a small procedure is used only a few times, it makes sense to treat it as inlined at the point of call for the purpose of analysis---either inlining it for real into the control-flow graph of the caller procedure, or descending into the call during analysis in a way that simulates inlining.
We however need to really deal with the pushdown structure when there are (possibly recursive) procedures.
We propose running first the analyses of \citet{Touzeau:2019:FEA:3302515.3290367} on the flattened structure, disregarding the results of return edges when they would be used to establish the existence of an execution.
For the interval analysis, this means taking $l'_a=+\infty$ and $h'_a=0$ when going through a return edge;
for the exact analysis on the flattened structure, this means that diagnoses ``$a$ is always in the cache at this location'' and ``$a$ is never in the cache at this location'' can be retained, but not diagnoses ``$a$ is definitely unknown as this location''.

We thus propose using the backtracking algorithm of Section~\ref{sec:backtracking} only on the cases that still have not been classified by the above approach.

\section{Presentation of other policies}
\emph{In this section, for completeness, we recall facts from {\JACMt}. The details of the policies are not important in our proof, since we will reuse many of their reduction arguments.}

The FIFO policy, also known as ``round-robin'', stores block according to their age, but a crucial difference is that the age is not that of the most recent access to the block, but of its entrance to the cache: in contrast to LRU, a block is not rejuvenated if it is accessed when it is already in the cache.
This allows implementing it simply using an index into a circular buffer.
For instance, if a $\nWays$-way cache contains $abcd$ in ascending age and $b$ is accessed, the cache remains $abcd$; if $e$ is accessed, it then contains $eabc$.

FIFO however has worse practical performance than LRU \citep{Al-Zoubi:2004:PEC:986537.986601}. ``Pseudo LRU'' replacement policies, meant to do in practice what LRU does (evict blocks that were not used recently), were thus proposed, in particular \citep{Al-Zoubi:2004:PEC:986537.986601,Reineke_PhD}:
\begin{itemize}
\item one based on a tree of direction bits leading to cache lines, named ``Tree PLRU'', ``PLRU-t'', or simply, as in this paper, ``PLRU'';
\item one based of ``most recently used'' bits, named ``PLRU-m'', ``MRU'' or, as in this paper, ``NMRU'' \citep{MRU_patent}.
\end{itemize}

Pseudo-LRU policies yield comparable practical performance to LRU \citep{Al-Zoubi:2004:PEC:986537.986601}, but they are very different from the point of view of static analysis.
For instance, a PLRU cache may, with a specifically concocted cache access pattern, indefinitely retain some data that was used only once at the beginning of execution and is never accessed again \citep{DBLP:journals/pieee/HeckmannLTW03}, which of course cannot happen with LRU.
This results in domino effects: the cache behavior of a loop body may be indefinitely affected by the cache contents before the loop~\citep{berg:OASIcs:2006:672}.
All of this makes the static analysis of programs over PLRU caches difficult; there are no known precise and fast analyses for this policy and for NMRU.

\paragraph{PLRU}
The cache lines of a PLRU cache, which may contain cached blocks, are arranged as the leaves of a full binary tree --- thus the number of ways $\nWays$ is a power of $2$, often $4$ or~$8$.
Two lines may not contain the same block.
Each internal node of the tree has a tag bit, which is represented as an arrow pointing to the left or right branch.
The state of the cache is thus the content of the lines and the $\nWays-1$ tag bits.

There is always a unique line such that there is a sequence of arrows from the root of the tree to the line; this is the line \emph{pointed at by the tags}.
Tags are said to be \emph{adjusted away} from a line as follows: on the path from the root of the tree to the line, tag bits are adjusted so that the arrows all point away from that path.

When a block $a$ is accessed:
\begin{itemize}
\item If the block is already in the cache, tags are adjusted away from this line.
\item If the block is not already in the cache and one or more cache lines are empty, the leftmost empty line is filled with $a$, and tags are adjusted away from this block.
\item If the block is not already in the cache and no cache line is empty, the block pointed at by the tags is evicted and replaced with $a$, and tags are adjusted away from this block.
\end{itemize}

\paragraph{NMRU}
The state of an $\nWays$-way NMRU cache is a sequence of at most $\nWays$ memory blocks $\alpha_i$, each tagged by a $0/1$ ``MRU-bit'' $r_i$ saying whether the associated block is to be considered not recently used ($0$) or recently used ($1$),
denoted by $\alpha_1^{r_1} \dots \alpha_{\nWays}^{r_\nWays}$.

An access to a block in the cache, a \emph{hit}, results in the associated MRU-bit being set to~$1$. If there were already $\nWays-1$ MRU-bits equal to $1$, then all the other MRU-bits are set to~$0$.

An access to a block $a$ not in the cache, a \emph{miss}, results in:
\begin{itemize}
\item if the cache is not full (number of blocks less than $\nWays$), then $a^1$ is appended to the sequence
\item if the cache is full (number of blocks equal to $\nWays$), then the leftmost (least index~$i$) block with associated MRU-bit $0$ is replaced by~$a^1$.
If there were already $\nWays-1$ MRU-bits equal to $1$, then all the other MRU-bits are set to~$0$.
\end{itemize}

\section{EXPTIME-completeness for Boolean register machines with procedure calls}
\citet{Touzeau:2019:FEA:3302515.3290367} showed PSPACE-hardness for PLRU, NMRU and FIFO analysis problems by providing, for each policy, a way to simulate a ``Boolean register machine''.
We extend this definition with procedure calls and prove EXPTIME-completeness.

\subsection{Extension to procedure calls}\label{part:procedure_calls}

\begin{definition}
A \emph{Boolean register machine} \citep{Monniaux:2019:CCA:3368192.3366018} is defined by a number $\nRegs$ of registers and a directed (multi)graph with an initial node and a final node, with edges adorned by instructions of the form:
\begin{description}
\item[Guard] $v_i=b$ where $1 \leq i \leq \nRegs$ and $b \in \{\false, \true\}$,
\item[Assignment] $v_i:=b$ where $1 \leq i \leq \nRegs$ and $b \in \{\false, \true\}$.
\end{description}

The \emph{register state} is a vector of $\nRegs$ Booleans.
An edge with a guard $v_i=b$ may be taken only if the $i$-th register contains~$b$; the register state is unchanged.
The register state after the execution of an edge with an assignment $v_i:=b$ is the same as the preceding register state except that the $i$-th register now contains~$b$.

The \emph{reachability problem} for such a system is the existence of a valid execution starting in the initial node with all registers equal to $\false$, and leading to the final node.
\end{definition}

\citet[Lemma~20]{Monniaux:2019:CCA:3368192.3366018} show:
\begin{lemma}
The reachability problem for Boolean register machines is PSPACE-complete.
\end{lemma}

\citet{DBLP:conf/tacas/GodefroidY13} introduced \emph{Boolean programs} or \emph{extended recursive state machines}, which are essentially Boolean register machines except that:
\begin{enumerate}
\item they allow procedure calls
\item they have local variables
\item they allow arbitrary transition relations with arbitrary guard predicates and arbitrary commands assigning new values to registers using arbitrary Boolean functions of the current values of registers.
\end{enumerate}

They prove the following theorem, for which we give an alternative proof in \autoref{part:exptime_complete}:
\begin{theorem}\label{thm:exptime_complete}
The reachability problem in Boolean programs is EXPTIME-complete.
\end{theorem}

Their result, however, establishes hardness in too generic a class of programs for our purposes: we want neither local variables nor arbitrary transitions. The following two lemmas get rid of them.

\begin{lemma}\label{rmk:no_complex_transitions}
Reachability in a Boolean program using arbitrary transitions reduces, in polynomial time, to reachability in a Boolean program only using constant guards and assignments.
\end{lemma}

\begin{proof}
{\JACMt} remark that arbitrary transitions defined by Boolean functions can be simulated only using ``guard'' and ``assignment'' elementary operations, with only linear blowup.
To each subterm $s$ of the expressions in the transition we associate a register variable:
\begin{itemize}
\item variables $b_i$ or $b'_i$ are retained;
\item for the result $r$ of each operator, a fresh variable is created.
\end{itemize}
Then, the truth table of each operator is encoded: for an operator $r=\textit{op}(x_1,\dots,x_n)$ ($n \leq 2$ for conventional Boolean operators), for each of the $2^n$ possible choices of the inputs, a sequence of guards on $x_1,\dots,x_n$ keeps only that choice, followed by an assignment to $r$ of the correct value; a nondeterministic choice is made between all these inputs.
(Several choices of inputs can be collapsed into the same sequence, if possible.)
\end{proof}

\begin{example}\label{ex:encoding_gates}
  Consider a transition from $(b_1,b_2)$ to $(b'_1,b'_2)$ defined as $(b'_1=b_1) \land (b'_2=b_1 \land \neg b_2)$.
  Create a fresh variable $r$ standing for the result of $\neg b_2$.
  The $b'_1:=b_1$, $r := \neg b_2$ and $b'_2 := b_1 \land r$ operations are encoded into:
  
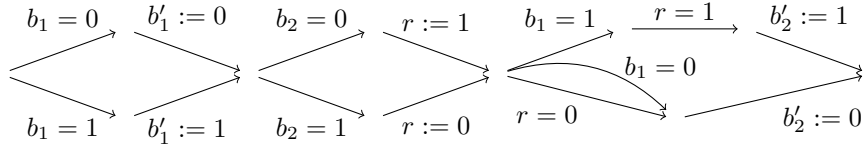
\begin{equation*}
  \begin{tikzpicture}[node distance=1em and 4em, ->,auto]
    \node(qs) {};
    \node(qs0) [above right= of qs] {};
    \node(qs1) [below right= of qs] {};
    \node(qsr) [above right= of qs1] {};
    \path(qs) edge node[above=1ex] {$b_1 = 0$} (qs0);
    \path(qs0) edge node[above=1ex] {$b'_1 :=  0$} (qsr);
    \path(qs) edge node[below=1ex] {$b_1 = 1$} (qs1);
    \path(qs1) edge node[below=1ex] {$b'_1 :=  1$} (qsr);
    
    \node(q0) [above right= of qsr] {};
    \node(qfr) [below right= of q0] {};
    \node(q1) [below right= of qsr] {};
    \path(qsr) edge node[above=1ex] {$b_2 = 0$} (q0);
    \path(q0) edge node[above=1ex] {$r :=  1$} (qfr);
    \path(qsr) edge node[below=1ex] {$b_2 = 1$} (q1);
    \path(q1) edge node[below=1ex] {$r :=  0$} (qfr);

    \node(qa1) [above right= of qfr] {};
    \node(qa2) [right= of qa1] {};
    \node(qf) [below right= of qa2] {};
    \node(qa0) [below right= 1em and 6em of qfr] {};
    \path(qfr) edge node[above=1ex] {$b_1=1$} (qa1);
    \path(qa1) edge node[above] {$r=1$} (qa2);
    \path(qa2) edge node[above=1ex] {$b'_2:=1$} (qf);
    \path(qfr) edge[bend left] node[right=2ex] {$b_1=0$} (qa0);
    \path(qfr) edge node[below left] {$r=0$} (qa0);
    \path(qa0) edge node[below right] {$b'_2:=0$} (qf);
  \end{tikzpicture}
  \end{equation*}
\end{example}

\begin{lemma}\label{rmk:no_local_variables}
Reachability in a Boolean program using local variables and only constant guards and assignments reduces, in polynomial time, to reachability in a Boolean program only using constant guards and assignments but no local variables.
\end{lemma}

\begin{proof}
  We essentially need a mechanism for saving registers on the stack at function entry and restoring them at function exit.

  A procedure $P$ is turned into $R+1$ procedures $P_0, \dots, P_R$ where $R$ is the number of registers $r_{i_1},\dots,r_{i_R}$ to save.
  $P_0$ is just $P$.
  Calls to $P$ are replaced by calls to $P_R$.
  Each procedure $P_j$ ($1 \leq j \leq R$) consists of two sequences, with nondeterministic choice between them:
  \begin{itemize}
  \item guard $r_{i_j} = 0$; call $P_{j-1}$; assignment $r_{i_j} := 0$;
  \item guard $r_{i_j} = 1$; call $P_{j-1}$; assignment $r_{i_j} := 1$.
  \end{itemize}
\end{proof}

\begin{example}
  Procedure $P$ has local variables $r_5$ and $r_7$.
  We create procedures $P_1$ and $P_2$ as follows:
  \begin{center}
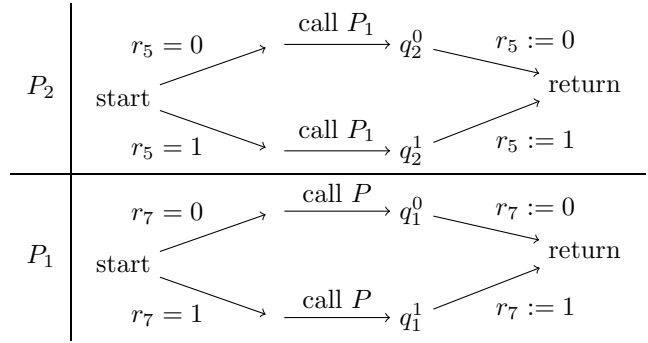
  \begin{tabular}{l|l}
    $P_2$ &
       \begin{tikzpicture}[baseline={([yshift=-.5ex]current bounding box.center)},node distance=1em and 4em, ->,auto]
         \node(qs) { start };
         \node(q0s) [above right=of qs] { };
         \node(q0c) [right=of q0s] { $q_2^0$ };
         \node(q1s) [below right=of qs] { };
         \node(q1c) [right=of q1s] { $q_2^1$ };
         \node(qf) [above right=of q1c] { return };
         \path(qs) edge node [above left] {$r_5 = 0$ } (q0s);
         \path(q0s) edge node [above] {call $P_1$} (q0c);
         \path(q0c) edge node [above right] {$r_5 := 0$ } (qf);
         \path(qs) edge node [below left] {$r_5 = 1$ } (q1s);
         \path(q1s) edge node [above] {call $P_1$} (q1c);
         \path(q1c) edge node [below right] {$r_5 := 1$ } (qf);
       \end{tikzpicture} \\
    \hline
    $P_1$ &
       \begin{tikzpicture}[baseline={([yshift=-.5ex]current bounding box.center)},node distance=1em and 4em, ->,auto]
         \node(qs) { start };
         \node(q0s) [above right=of qs] { };
         \node(q0c) [right=of q0s] { $q_1^0$ };
         \node(q1s) [below right=of qs] { };
         \node(q1c) [right=of q1s] { $q_1^1$ };
         \node(qf) [above right=of q1c] { return };
         \path(qs) edge node [above left] {$r_7 = 0$ } (q0s);
         \path(q0s) edge node [above] {call $P$} (q0c);
         \path(q0c) edge node [above right] {$r_7 := 0$ } (qf);
         \path(qs) edge node [below left] {$r_7 = 1$ } (q1s);
         \path(q1s) edge node [above] {call $P$} (q1c);
         \path(q1c) edge node [below right] {$r_7 := 1$ } (qf);
       \end{tikzpicture} \\
    
  \end{tabular}
\end{center}
Essentially, we use the call stack to store the value of $r_5$ (encoded into a return control location $q_2^0$ or $q_2^1$ depending on its value), then the value of $r_7$ (encoded into a return control location $q_1^0$ or $q_1^1$ depending on its value).
\end{example}

We have gotten rid of the arbitrary transitions and the local variables.
Let us now proceed with the rest of the reductions.
{\JACMt} prove the following:
\begin{itemize} 
\item  The reachability problem on Boolean register machines with acyclic control flow graph is NP-complete.
\item The reachability problem on Boolean register machines is PSPACE-complete.
\end{itemize}
and we will similarly prove that
the reachability problem on Boolean programs where the only transitions are constant guards and constant assignments and without local variables is EXPTIME-complete.

Boolean programs where the only transitions are constant guards and constant assignments and without local variables are, equivalently, Boolean register machines with procedure calls:

\begin{definition}
  A \emph{Boolean register machine with procedure calls} is defined by a number $\nRegs$ of registers, a number $P \geq 1$ of procedures, and $P$ directed (multi)graphs, called \emph{procedures}, with an initial node and a final node, with edges adorned by instructions of the form:
\begin{description}
\item[Guard] $v_i=b$ where $1 \leq i \leq \nRegs$ and $b \in \{\false, \true\}$,
\item[Assignment] $v_i:=b$ where $1 \leq i \leq \nRegs$ and $b \in \{\false, \true\}$,
\item[Call] $\textit{call}(i)$ where $1 \leq i \leq P$.
\end{description}

A \emph{control location} in such a machine is a pair $(i,j)$ where $1 \leq i \leq P$ is the index of a procedure and $j$ is a control vertex inside procedure~$i$.
The \emph{configuration} of a Boolean register machine with procedure calls consists of a control location, the state of the $\nRegs$ registers, and a call stack, a (possibly empty) sequence of control locations.
The execution starts at vertex $1$ of procedure $1$ with an empty stack and zeroes in the registers.
When a procedure $i$ with $N_i$ control locations is called, its execution starts at location $(i,1)$ and stops at location $(i,N_i)$; a location is then popped from the stack and control is transferred to it.

A \emph{reachability problem} in such a machine is whether a given control location is reachable.
\end{definition}

\begin{lemma}
  The reachability problem for Boolean register machines with procedure calls lies in EXPTIME.
\end{lemma}

\begin{proof}
\autoref{thm:exptime_complete} states membership in EXPTIME for the more general case with local variables and arbitrary transitions.
\end{proof}

\begin{theorem}
  The reachability problem for Boolean register machines with procedure calls is EXPTIME-complete.
\end{theorem}

\begin{proof}
Compose the reductions of \autoref{rmk:no_complex_transitions}, \autoref{rmk:no_local_variables} and \autoref{thm:exptime_complete} to establish hardness.
\end{proof}

\section{EXPTIME-completeness for non-LRU policies}
We show EXPTIME-hardness by reducing arbitrary reachability problems on Boolean register machines with procedure calls to ``exist-hit'' and ``exist-miss'' problems.
We reuse {\JACMt}'s encoding of the Boolean registers into the cache state, and their transformation of Boolean register machines into control-flow graphs adorned with cache blocks, suitable for cache analysis.
This transformation retains the structure of the control-flow, replacing each instruction edge from the Boolean register machine by a ``gadget'' making cache accesses.
To deal with multiple procedures in our problems, we translate each procedure independently and retain call instructions.

\subsection{Encoding for programs without procedures}
{\JACMt} considered programs without procedures.
For each replacement policy, they have
\begin{itemize}
\item a notion of \emph{well-phased} cache state: the initial cache state is well-phased and all gadgets preserve well-phasedness (that is, at their boundary: they use not well-phased states inside the gadget);
\item
  a notion of a \emph{well-formed} cache state \emph{corresponding} to a register state, meaning it encodes that state; only well-formed cache states are meaningful for the transformation of the reachability problems;
\item well-phased but not well-formed cache states may only lead to further well-phased but not well-formed cache states;
  well-phased but not well-formed cache states appear in valid execution traces of the cache analysis problems that are not meaningful for the reduction.
\end{itemize}

Their reductions turn a Boolean register machine into a control flow graph with edges adorned with cache blocks, and thus a reachability problem into a cache analysis problem, as follows:
\begin{itemize}
\item A prologue, set at the entry point of the control flow graph, suitably initializes the cache contents;
\item the main part of the control-flow graph is identical to the Boolean register machine where each instruction (guard or assignment) is replaced by a ``gadget'': a piece of control-flow graph adorned with cache accesses;
  the gadget simulates on the cache state what happens to the Boolean registers;
  if a guard fails, the gadget stops execution or leads to well-phased but not well-formed cache content;
\item an epilogue, leading to the exit point of the control-flow graph, filters out well-phased but not well-formed cache content and prepares the cache so that the exist-miss or exist-hit problem at the control-flow graph exit points exactly answers the reachability problem for the Boolean register machine.
\end{itemize}
The prologue is different if the initial cache is empty or has arbitrary content, and the epilogue is different for exist-miss and exist-hit problems.

Their results can be summarized as: the execution traces of the Boolean register machine reaching the final state of that machine are in a one-to-one correspondence with the execution traces of the control-flow graph that reach the exist-hit (respectively, exist-miss) condition at the end:
the execution sequence of edges in the Boolean register machine maps to a sequence of ``gadgets'' (prologue, then one gadget per edge of the Boolean register machine, then epilogue).
The control-flow graph labeled with cache blocks may have other executions, but they create not well-formed cache content and thus are ignored by the condition in the epilogue.

The encoding, the notions of well-phasedness and well-formedness, and the gadgets used, are completely different for each policy.
We refer readers to {\JACMt} for more details, and shall here just sketch how they encode the reachability problem for a Boolean register machine to the exist-hit problem for the FIFO cache.
The associativity of the cache is chosen as $\nWays = 2\nRegs-1$.
The alphabet of cache blocks is
$\{ (a_{i,b})_{1 \leq i \leq \nRegs,b \in \{\false,\true\}} \} \cup
\{ (e_i)_{1 \leq i \leq \nRegs} \} \cup
\{ (f_i)_{1 \leq i \leq \nRegs-1} \} \cup
\{ (g_i)_{1 \leq i \leq \nRegs-1} \}$.

The main idea is to encode the value of registers by loading the blocks $a_{i,b}$ into the cache ($a_{i,\true}$ is used when the register $i$ contains $1$, and $a_{i,\false}$ is used for~$0$).
The blocks $e_{i}$ are used to distinguished valid Boolean machine executions from executions where the machine should have halted.
Finally, blocks $f_i$ and $g_i$ are used in epilogue to turn valid states into cache hits and invalid states into cache misses.

The register state $v_1,\dots,v_{\nRegs}$ of the register machine is to be encoded as the FIFO state, acting essentially as a delay-line memory:
\begin{equation}
  a_{1,v_1} e_2 a_{2,v_2} \dots e_{\nRegs} a_{\nRegs,v_{\nRegs}}.
\end{equation}
FIFO states that are not of this form are considered not well-formed.

The register machine graph is turned into a cache analysis graph as follows.
\begin{itemize}
\item From the cache analysis initial node $I_f$ to the register machine former initial node $I_r$ there is a prologue, a sequence of accesses $a_{1,\false} e_2 \dots a_{\nRegs-1,\false} e_{\nRegs} a_{\nRegs,\false}$.

\item Each guard edge $v_i = b$ is replaced by the gadget
  
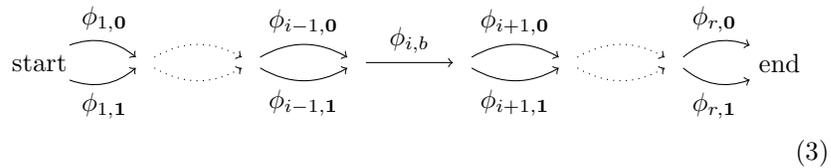
\begin{equation}
    \begin{tikzpicture}[node distance=4em,->]
      \node (q0) {start};
      \node (q1) [right of=q0] {};
      \node (qim2) [right of=q1] {};
      \node (qim1) [right of=qim2] {};
      \node (qi) [right of=qim1] {};
      \node (qip1) [right of=qi] {};
      \node (qRm1) [right of=qip1] {};
      \node (qR) [right of=qRm1] {end};
      \path (q0) edge [bend left] node[above] {$\phi_{1,\false}$} (q1);
      \path (q0) edge [bend right] node[below] {$\phi_{1,\true}$} (q1);
      \path (q1) edge [dotted, bend left] (qim2);
      \path (q1) edge [dotted, bend right] (qim2);
      \path (qim2) edge [bend left] node[above] {$\phi_{i-1,\false}$} (qim1);
      \path (qim2) edge [bend right] node[below] {$\phi_{i-1,\true}$} (qim1);
      \path (qim1) edge node[above] {$\phi_{i,b}$} (qi);
      \path (qi) edge [bend left] node[above] {$\phi_{i+1,\false}$} (qip1);
      \path (qi) edge [bend right] node[below] {$\phi_{i+1,\true}$} (qip1);
      \path (qip1) edge [dotted, bend left] (qRm1);
      \path (qip1) edge [dotted, bend right] (qRm1);
      \path (qRm1) edge [bend left] node[above] {$\phi_{\nRegs,\false}$} (qR);
      \path (qRm1) edge [bend right] node[below] {$\phi_{\nRegs,\true}$} (qR);
    \end{tikzpicture}
  \end{equation}
  where $\phi_{i,b}$ denotes the sequence of accesses
  $a_{i,b} e_i a_{i,b}$.

\item Each assignment edge $v_i := b$ is replaced by the gadget
  
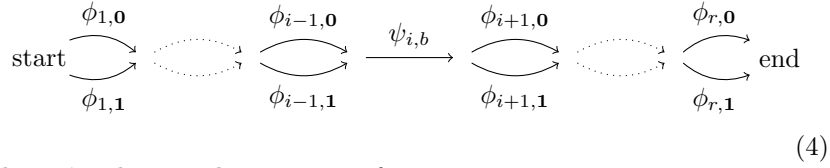
\begin{equation}
    \begin{tikzpicture}[node distance=4em,->]
      \node (q0) {start};
      \node (q1) [right of=q0] {};
      \node (qim2) [right of=q1] {};
      \node (qim1) [right of=qim2] {};
      \node (qi) [right of=qim1] {};
      \node (qip1) [right of=qi] {};
      \node (qRm1) [right of=qip1] {};
      \node (qR) [right of=qRm1] {end};
      \path (q0) edge [bend left] node[above] {$\phi_{1,\false}$} (q1);
      \path (q0) edge [bend right] node[below] {$\phi_{1,\true}$} (q1);
      \path (q1) edge [dotted, bend left] (qim2);
      \path (q1) edge [dotted, bend right] (qim2);
      \path (qim2) edge [bend left] node[above] {$\phi_{i-1,\false}$} (qim1);
      \path (qim2) edge [bend right] node[below] {$\phi_{i-1,\true}$} (qim1);
      \path (qim1) edge node[above] {$\psi_{i,b}$} (qi);
      \path (qi) edge [bend left] node[above] {$\phi_{i+1,\false}$} (qip1);
      \path (qi) edge [bend right] node[below] {$\phi_{i+1,\true}$} (qip1);
      \path (qip1) edge [dotted, bend left] (qRm1);
      \path (qip1) edge [dotted, bend right] (qRm1);
      \path (qRm1) edge [bend left] node[above] {$\phi_{\nRegs,\false}$} (qR);
      \path (qRm1) edge [bend right] node[below] {$\phi_{\nRegs,\true}$} (qR);
    \end{tikzpicture}
  \end{equation}
  where $\psi_{i,b}$ denotes the sequence of accesses $e_i a_{i,b} e_i$.
\item From the register machine former final node $F_r$ to a node $F_a$ there is a sequence of accesses
  $\psi_{1,\false} \dots \psi_{\nRegs,\false}$,
  constituting the first part of the epilogue.
\item From $F_a$ to a node $F_h$ there is a sequence of accesses 
  \begin{equation*}
  a_{1,\false} g_1 e_2 f_2 a_{2,\false} g_2 \dots e_{\nRegs-1} f_{\nRegs-1} a_{\nRegs-1,\false} g_{\nRegs-1} e_{\nRegs} f_{\nRegs},
  \end{equation*}
  constituting the second part of the epilogue.
\item The final node is $F_f = F_h$.
\end{itemize}

The main difficulty in this reduction is that the Boolean register machines may terminate traces if a guard is not satisfied, whereas the cache problem has no guards and no way to terminate traces.
The workaround is that cache states that do not correspond to traces from the Boolean machine are irremediably marked as not well-formed: they may lead only to more not well-formed states.

The encoding is chosen such that:
\begin{lemma}
  Assume starting in a well-formed FIFO state, corresponding to state $\sigma$, then any path through the gadget encoding an assignment or a guard
  \begin{itemize}
  \item either leads to a well-formed FIFO state, corresponding to the state $\sigma'$ obtained by executing the assignment, or $\sigma'=\sigma$ for a valid guard;
  \item or leads to a not well-formed state.
  \end{itemize}
\end{lemma}

\begin{lemma}
  Assume starting in a not well-formed state, then any path through the gadget encoding an assignment or a guard leads to a not well-formed state.
\end{lemma}

\begin{corollary}
  Any path from a well-formed FIFO state in $I_r$ to $F_r$ in the FIFO graph
  \begin{itemize}
  \item either corresponds to a valid sequence of assignments and guards from the register machine from $I_r$ to $F_r$, and leads to a well-formed FIFO state corresponding to the final state of that sequence
  \item or corresponds to an invalid sequence of assignments and guards from the register machine (some guards were no satisfied), and leads to a not well-formed FIFO state.
  \end{itemize}

  Conversely, any valid sequence of assignments and guards from the register machine maps from $I_r$ to $F_r$ transforms a well-formed FIFO state into a well-formed FIFO state, corresponding respectively to the initial and final states of that sequence.
\end{corollary}

The epilogue is chosen so that it recognizes only correct states, whose encodings produce a cache hit, while not well-formed states lead to a cache miss: 
\begin{theorem}
There is an execution of the FIFO cache from $I_f$ to $F_f$ such that $a_{\nRegs,\false}$ is in the final cache state if and only if there is an execution of the Boolean register machine from $I_r$ to $F_r$.
\end{theorem}

\subsection{Extension to programs with procedure calls}
We extend our control flow graphs labeled with cache blocks with procedure calls, mimicking the procedure calls on the Boolean register machines, and we apply the same reduction: guards and assignments are replaced by gadgets, prologue and epilogue are added;
the difference is that we deal with procedure calls, which are kept intact.

The execution traces of the Boolean register machine with procedure calls reaching the final state of that machine are in a one-to-one correspondence with the execution traces of the control-flow graph with procedure calls, with control edges adorned by cache accesses, that reach the exist-miss (or exist-miss) condition at the end.

This reduction proves EXPTIME-hardness of the exist-miss and exist-hit problems for control-flow graphs with cache block accesses and procedure calls for the same cases as {\JACMt} proves PSPACE-hardness without procedure calls.

Membership in EXPTIME is easy to establish by reduction to Boolean programs, for which reachability properties are known to be in EXPTIME.
Indeed, a cache has an internal state, which can be encoded as a vector of bits (as in hardware): for FIFO, it consists of $\nWays$ block labels (if there are $n$ blocks in the system, each label takes $\lceil \log_1(\nWays-1) \rceil +1$ bits); for PLRU it consists of the block labels and the direction of the arrows in the PLRU tree; for NMRU it consists of the block labels and MRU bits.
For each of these policies, the effect of an access on the cache is implemented by a simple program using comparisons, assignments, etc.: if the block is in the cache, refresh it, if it is not, evict a block from the cache according to the policy and load the block.
This simple program can be expanded (as in the real hardware caches) into logical gates operating on the cache state, taking as input also the binary encoding of the label of the block being accessed.
The number of these logical gates is polynomial in $\nWays$ and in the number $n$ of distinct blocks in the cache analysis problem.
Logical gates can be encoded into Boolean program guards and assignments, as in Example~\ref{ex:encoding_gates}.

Putting all together, we prove:

\begin{theorem}
  The exist-miss and exist-hit problems are EXPTIME-complete for PLRU and FIFO caches for both empty cache and arbitrary cache initialization and control-flow graphs with procedure calls.
  The exist-miss and exist-hit problems are EXPTIME-complete for NMRU caches for empty cache initialization and control-flow graphs with procedure calls.
\end{theorem}

\section{Conclusion and perspectives}
Our work is yet another indication that LRU caches are easier to analyze statically and thus more suitable for applications where it is important to have static cache analysis---those requiring justifiable bounds on worst-case execution time, and possibly in security and cryptography where one must not leak information through a cache side channel.

One may object that our work, as that of {\JACMt}, establishes the hardness of unrealistic cases, with unbounded associativity and convoluted access patterns.
They already addressed this objection: with bounded associativity, the problems can be polynomially expanded and solved (the exponential is in the associativity), so they cannot be distinguished using the usual complexity classes; and attempting to establish asymptotic differences in polynomial degrees also leads to difficulties.


We have proposed a backtracking algorithm, based on the NP structure, for solving the analysis problems in the case of the LRU policy, as well as an approach for using this algorithm only as a last resort when some approximations have failed to resolve certain cases.

Are there practically efficient algorithms for solving analysis problems for the other policies, for which we proved the analysis to be EXPTIME-complete? Our proof of EXPTIME membership is basically ``expand exponentially the problem into one we know how to solve in polynomial time'', which is obviously explosive. Would there be lazy approaches to this expansion, leading to tolerable execution time and space on practical instances?

\printbibliography

\appendix

\section{Alternative proof for EXPTIME-completeness of reachability in Boolean programs}
\label{part:exptime_complete}
\citet{DBLP:conf/tacas/GodefroidY13} claim EXPTIME-hardness for reachability in Boolean programs (\autoref{thm:exptime_complete}), but they refer the reader to a full version of their article, which is available only by request to the authors. We thus provide, in the next subsections, an independent proof of EXPTIME-hardness for Boolean programs.

Note that EXPTIME membership is easily established. A Boolean register machine with procedure calls may be expanded into an equivalent pushdown system,
at the cost of exponential blowup: just consider one control location in the pushdown automaton for each control location in the Boolean register machine and each of the (exponentially many) vector of values of the registers;
then apply \autoref{th:reachability_pushdown_polynomial}.

\subsection{Succinctly represented problems}
We have seen how a reachability problem involving Boolean registers can be expanded into a reachability problem not involving registers, that is, a reachability problem in an oriented graph at the cost of exponential blowup.
This is an instance of a more general pattern relating the complexity of problems when they are represented as explicit lists of transitions versus ``implicit'' list of transitions, for instance involving registers, in the same way that a small Boolean formula is a succinct representation for a much larger explicit truth table.

\citet{DBLP:journals/iandc/GalperinW83} studied the complexity of various problems on graphs when these graphs are \emph{succinctly} represented, by which they mean that graph vertices are labeled by a vector of bits, and the adjacency relation is defined by a Boolean circuit taking as inputs two vectors of bits and answering one bit: whether the vertices labeled by these two vectors are connected.
\citet{DBLP:journals/iandc/PapadimitriouY86} generalized their results:
a NP-complete problem (respectively, P-complete; NLOGSPACE-complete) problem on explicitly represented graphs, 
under some fairly permissive condition on the reduction used for showing this completeness property,
becomes NEXPTIME-complete (respectively, EXPTIME-complete; PSPACE-complete) on succinctly represented graphs.
A well-known example of this phenomenon is the reachability problem: given two vertices in a directed graph, say whether one is reachable from another---it is NLOGSPACE-complete on explicitly represented graphs, and becomes PSPACE-complete on succinctly represented graphs, where it is also known as the reachability problem in implicit-state model checking.

The reachability problem for explicitly represented pushdown systems, which are very close to Boolean register machines with procedure calls but no registers, is known to be P-complete.
We can thus hope that it becomes EXPTIME-complete for succinctly represented pushdown systems; however we cannot use \citeauthor{DBLP:journals/iandc/PapadimitriouY86}' results because they pertain solely to graph problems.
We can however follow the same general approach as their hardness proof: analyze the reduction from the acceptance problem for polynomial-time Turing machines to the problem for explicitly represented pushdown systems, which are close to Boolean programs without registers, and construct a reduction from the acceptance problem for exponential-time Turing machines to the problem for succinctly represented pushdown systems, which are close to Boolean programs with registers.

It takes four reduction steps to show that the reachability problem for Boolean register machines with procedure calls and $0$ registers is P-hard:
\begin{inparaenum}[(i)]
\item from the acceptance problem for polynomial-time Turing machines to the circuit value problem (CVP) \cite[4.2]{Greenlaw_Hoover_Ruzzo_P_completeness}
\item from the CVP to the monotone circuit value problem \cite[A.1.3]{Greenlaw_Hoover_Ruzzo_P_completeness}
\item from the monotone CVP to the emptiness problem for context-free grammars \cite[A.7.2]{Greenlaw_Hoover_Ruzzo_P_completeness}
\item from emptiness in context-free grammars to reachability in Boolean programs with local variables.
\end{inparaenum}

\subsection{Reductions for explicit descriptions}
The circuit value problem (CVP) is: given a Boolean circuit, using logical gates $\land$, $\lor$, $\neg$, with known inputs, compute its output.
The first reduction step \citep{10.1145/990518.990519} \citep[Th.~4.2.2]{Greenlaw_Hoover_Ruzzo_P_completeness} encodes the bounded deterministic execution of a Turing machine into a circuit in much the same way that one encodes the bounded nondeterministic execution of a Turing machine into a Boolean satisfiability problem: the value $c_{i,j}$ of each cell at each position $j$ in the tape at each point in time $i > 0$ is defined as a function of $c_{i-1,j-1}$, $c_{i-1,j}$ and $c_{i+1,j}$, with a different value whether the read/write head is on the cell; then these values $c_{i,j}$ are encoded into a vector of bits (of size logarithmic in the size of the tape alphabet and the number of control states), and one then obtains a circuit.
It then suffices to add initialization for values $c_{0,j}$ of the cells at time $0$, and a test for a reachability condition.

The monotone CVP is: given a Boolean circuit, using logical gates $\land$ and $\lor$ with known inputs, compute its output.
Obviously it is a subset of the general CVP.
A general CVP can be encoded into a monotone CVP by using ``dual rail encoding''  \citep{10.1145/1008354.1008356} \citep[Th.~6.2.2]{Greenlaw_Hoover_Ruzzo_P_completeness}: each wire $b$ in the original circuit is encoded into two wires $b_0$ and $b_1$, where $b_0$ is $1$ if $b$ is $0$, $0$ if $b$ is $1$, and
$b_1$ is $1$ if $b$ is $1$, $0$ if $b$ is $0$.
It is possible to simulate each $\land$ or $\lor$ gate of the original circuit by two monotone gates; $\neg$ gates map to swapping of two wires.

Let us now encode the monotone CVP into the context-free grammar emptiness problem \citep[A.7.2, crediting Martin Tompa]{Greenlaw_Hoover_Ruzzo_P_completeness}.
To each wire $w_i$ in the circuit one associates a nonterminal $\nu_i$.
If $w_i$ is initialized to $1$, then we add a rule $\nu_i \rightarrow \emptyword$ (meaning that $\nu_i$ accepts the empty word; equivalently one may introduce a nonterminal $a$ and have a rule $\nu_i \rightarrow a$).
We add no rule if $w_i$ is initialized to $0$.
If $w_i$ is defined as $w_j \lor w_k$, then we add two rules $\nu_i \rightarrow \nu_j$ and $\nu_i \rightarrow \nu_k$.
If $w_i$ is defined as $w_j \land w_k$, then we add a rule $\nu_i \rightarrow \nu_j \nu_k$.
The nonterminal $\nu_1$ to test for emptiness is the one that corresponds to the output wire of the monotone circuit.

Finally, let us encode the context-free grammar emptiness problem into the reachability problem for a Boolean program without registers.
This is the well-known relationship between context-free grammars and procedure calls in structured programs.
Each nonterminal in the grammar becomes a procedure.
A derivation rule $L \rightarrow R_1 \dots R_n$ becomes a sequence of calls to th procedures corresponding to nonterminals $R_1$ to $R_n$, starting in the initial control location of the procedure associated with nonterminal $L$ and ending in the final location of that procedure.

\subsection{Lifting reductions to implicitly represented problems}
In the above reductions, circuits are described as a list of gates.
The first reduction step, from Turing machines to CVP, is however highly repetitive: the same construction is applied for all $i > 0$ and $j$.
We thus use the notion of \emph{succinctly described circuit} \cite[ch.~20]{Papadimitriou_complexity}:
wires $w_i$ are identified by their index $i$ written in binary,
and gates in the succinctly represented circuits are introduced by rules of the form
$C(i,j,k): w_i = w_j \land w_k$,
$C(i,j,k): w_i = w_j \lor w_k$,
$C(i,j): w_i = \neg w_j$,
where $C$ is a condition over the binary encodings of indices $i,j,k$, itself expressed as a Boolean circuit, that constrains for which indices the gate is created.
The notion of \emph{succinctly described monotone circuit} is defined similarly.

The encodings described above for turning a reachability problem on the execution of a polynomially bounded Turing machine into an explicitly described CVP of polynomial size, then into an explicitly described monotone CVP of polynomial size, can be applied to turn a reachability problem on the execution of an exponentially bounded Turing machine into a succinctly described CVP of polynomial size, then into a succinctly described monotone CVP of polynomial size.%
\footnote{Succinctly described circuits, and the EXPTIME-completeness of their value problem, have long been known~\cite[Ch.~20]{Papadimitriou_complexity}. We however recall how to establish this result for the sake of completeness and easier understanding of how we turn successive reductions for P-completeness for explicitly described problems into successive reductions for EXPTIME-completeness on succinctly described problems.}

We define similarly the notion of a succinctly represented context-free grammar.
A succinct rule $C(i,j,k): \nu_i \rightarrow \nu_j \nu_k$ (for arity $2$; other arities are similarly defined), where $C$ is a Boolean circuit over the binary encodings of $i$, $j$ and $k$, encodes a family of rules $\nu_i \rightarrow \nu_j \nu_k$ for all $i,j,k$ such that $C(i,j,k)$ returns $1$.
As with explicitly described monotone CVPs, a succinctly described monotone CVP can be transformed into a succinctly represented context-free grammar emptiness problem.

The variables $i$, $j$ etc. are binary encodings. For the final reduction to Boolean register machines with procedures, we put these Boolean encodings into the local variables of the Boolean programs.

\end{document}